\documentclass[journal]{IEEEtran}

\usepackage[T1]{fontenc}
\usepackage[utf8]{inputenc}
\usepackage{etoolbox}
\usepackage{microtype}
\usepackage{amsmath,amssymb,amsthm,mathtools,bm}
\usepackage[nocompress]{cite}
\usepackage{graphicx}
\usepackage{xcolor}
\usepackage{subfigure}
\usepackage{hyperref}
\usepackage[nameinlink,noabbrev]{cleveref}
\hypersetup{colorlinks=true,linkcolor=black,citecolor=black,urlcolor=black}

\makeatletter
\def\endthebibliography{%
  \def\@noitemerr{\@latex@warning{Empty `thebibliography' environment}}%
  \endlist}
\makeatother

\makeatletter
\newcommand{\centeredfigurecaption}[2]{%
  \begingroup
  \long\def\@makecaption##1##2{%
    \@IEEEfigurecaptionsepspace
    \hbox to\hsize{\normalfont\footnotesize\hfil
      \hbox{\normalfont\footnotesize{##1.}\nobreakspace\nobreakspace ##2}%
      \hfil}}%
  \caption{#1}\label{#2}%
  \endgroup
}
\makeatother

\newtheorem{theorem}{Theorem}
\newtheorem{lemma}{Lemma}
\newtheorem{proposition}{Proposition}
\newtheorem{corollary}{Corollary}
\theoremstyle{definition}

\theoremstyle{plain}
\newtheorem{remark}{Remark}

\newcommand{\E}{\mathbb{E}}
\newcommand{\Kset}{\mathcal{K}}

\newcommand{\Tatt}{{T^{\mathsf{att}}}}

\newcommand{\Tcmm}{{T^{\mathsf{cmm}}}}
\newcommand{\Tcmp}{{T^{\mathsf{cmp}}}}
\newcommand{\Catt}{{C^{\mathsf{att}}}}
\newcommand{\Cffn}{{C^{\mathsf{ffn}}}}
\newcommand{\latt}{{\mu^{\mathsf{att}}}}
\newcommand{\lffn}{{\mu^{\mathsf{ffn}}}}
\newcommand{\Ps}{{P^{\mathsf{s}}}}

\newcommand{\Tsur}{{\tilde T}}
\newcommand{\Pbar}{{\bar P}}

\title{Low-Latency MoE Inference in Space Data Centers: A Wireless-Power-Transfer Perspective}
\author{Haotian~Zheng,
Zhanwei~Wang,
and Kaibin~Huang,~\IEEEmembership{Fellow,~IEEE}%
\thanks{H. Zheng, Z. Wang, and K. Huang are with the Department of Electrical and Computer Engineering, The University of Hong Kong, Hong Kong SAR, China (Email: \{htzheng, zhanweiw, huangkb\}@eee.hku.hk). Corresponding authors: K. Huang; Z. Wang.}}

\begin{document}
\bstctlcite{IEEErefControl}
\maketitle

\begin{abstract}

Space data centers (SpaceDCs) are emerging as platforms for large language models (LLMs), supported by continuous solar-energy harvesting and high-speed interconnections.
The mixture-of-experts (MoE) architecture supports distributed LLM inference by deploying experts across satellites.
However, skewed expert activation creates uneven computation power demands despite uniform solar-power supply at the dawn--dusk sun-synchronous orbit (SSO).
Consequently, active experts remain power-constrained while harvested power at idle experts is underused. These reduce inference efficiency and increase latency.
To address this supply--demand mismatch, we propose the SpaceWPT framework, which uses inter-satellite laser-based wireless power transfer (WPT) to redistribute harvested power.
The key design challenge is to balance the computation demands of power-donating and power-receiving satellites under WPT losses: transferring more power accelerates computation at the receiver but leaves less power for the donor.
After deriving design principles that simplify WPT control, we characterize the resulting latency tradeoff through a tractable surrogate of average inference latency.
Minimizing the surrogate yields a novel two-sided water-filling strategy that determines whether each active expert donates, retains, or receives power according to the water level and its WPT link efficiency.
Simulations with Mixtral-8x7B and Phi-3.5-MoE show up to 29.25\% lower inference latency than without WPT, with larger gains under more compact formations and higher transfer efficiencies.

\end{abstract}

\begin{IEEEkeywords}
Space data centers, Wireless Power Transfer, Mixture of Experts, Power Control.
\end{IEEEkeywords}

\section{Introduction}
\noindent

The \emph{sixth-generation} (6G) mobile networks are envisioned to extend ubiquitous intelligence beyond the terrestrial edge into space~\cite{6g_vision}.
Aligned with this vision, \emph{space data centers} (SpaceDCs), which are built on tight-formation satellite fleets with every satellite carrying a solar array and on-board GPUs, are emerging as a key platform for Space AI~\cite{liu2025computing_space,wang2026spaceai}.
Deployment in a dawn--dusk \emph{sun-synchronous orbit} (SSO) enables near-continuous and spatially uniform solar-energy harvesting across satellites, while inter-satellite distances of only hundreds of meters support high-bandwidth, low-latency connections among their GPUs~\cite{virgili2015descending,y2025towards}.
Together, these features make SpaceDCs a scalable, solar-powered computing platform for distributed \emph{large language models} (LLMs)~\cite{ict_investment}.
\emph{Mixture-of-Experts} (MoE) is a widely adopted architecture for LLMs, comprising multiple expert networks that can be deployed across satellites~\cite{wang2026spacemoe}.
Within this architecture, a gating mechanism activates experts according to the input token, and the activated experts perform computation in parallel~\cite{shazeer2017outrageously}.
However, token-dependent expert activation creates heterogeneous power demands that mismatch the uniform solar-power supply across the SpaceDC.
Under this \emph{supply--demand mismatch}, satellites with higher power demand remain constrained by their local solar-power supply.
This limits their computation speed and hence increases MoE inference latency.
To tackle this challenge, we propose SpaceWPT, a system that explores \emph{wireless power transfer} (WPT) in SpaceDCs to reallocate power among satellites to accelerate MoE inference.

Existing methods mitigate the supply-demand mismatch by redistributing power demand across nodes through model placement at deployment or expert routing during inference.
Model-placement strategies determine the power demand by assigning model components to computing nodes~\cite{mirhoseini2017device}.
Several works optimize MoE expert placement according to token workloads and activation patterns to balance processing loads and reduce inter-node communication~\cite{moetuner,prism}.
However,
inter-satellite link conditions in SpaceDCs vary continuously, causing the placement optimized at deployment to gradually become suboptimal.
Adjusting the placement to the varying links
requires transferring billion-parameter experts between satellites, incurring high communication overhead~\cite{team2026kimi}.
Alternatively, expert routing instead keeps model placement fixed and reshapes the token-level power demand by determining which experts process each token.
Existing methods balance expert loads by routing tokens to less-loaded experts when gate scores are comparable~\cite{laser}, or reduce communication overhead by replacing or skipping experts based on representation similarity and channel conditions~\cite{siftmoe}. 
Nevertheless, such operations alter the MoE model output and hence degrade inference accuracy.
The proposed WPT approach has no such limitations.
Rather than balancing power demand through workload redistribution, it adjusts the supply side by reallocating harvested solar power among satellites. 
This leaves expert placement and routing unchanged, avoiding the overhead of transferring expert parameters and preserving inference accuracy.

There exists a rich literature on the application of microwave-based WPT to power different types of terrestrial wireless networks, ranging from communication and sensor networks to edge computing networks.
Relevant studies use downlink WPT to maximize uplink throughput~\cite{ju2014wpcn} and deploy dedicated power beacons to optimize network deployment under outage constraints~\cite{huang2014cellularwpt}.
More recently, WPT has been integrated with edge computing and federated learning, where transferred energy powers local processing and wireless transmission to accelerate computation and communication at energy-constrained devices~\cite{you2016wpt,wpt_edge,zeng2022wpfeel}.
In contrast, laser-based WPT in space operates under different scales of inter-node distances and power demands.
Inter-satellite distances span hundreds of meters to kilometers, beyond the operating range of microwave-based WPT, while LLM inference requires substantially more power than low-power communication or sensing. 
These requirements favor \emph{laser-based WPT}, whose low-divergence laser beams maintain high power density over long distances and whose narrow spectrum enables photovoltaic receivers to be optimized for high conversion efficiency~\cite{mohsan2023owpt,fafard2021photovoltaic}.
Moreover, operating in space avoids the safety constraints that limit high-power laser transmission in terrestrial environments~\cite{soltani2022laser_safety}.
Existing laser-based WPT studies investigate transferring surplus energy from large satellites to energy-deficient CubeSats, small satellites used for communication, and mitigating pointing errors to improve harvested power~\cite{PowerHotspots,owpt_alignment}.
These works establish the feasibility of laser-based WPT in space, but they mainly focus on point-to-point single-link design.
However, WPT for distributed MoE in space is more sophisticated as it requires network-level coordinated designs across multiple satellites to accommodate their activation-based heterogeneous power demands.

Designing network-level WPT for distributed MoE inference poses two unique challenges in terms of controlling transfer directions and power levels.
First, transfer directions need to accommodate the uneven power demands induced by skewed expert activation.
As expert activation varies from token to token, the locations of power shortages and surpluses change, requiring corresponding adjustments to inter-satellite WPT.
Second, even with the transfer direction determined, the power level should be carefully controlled to balance the receiver’s power demand against the transmitter’s own computation needs.
Increasing the transfer power accelerates receiver-side computation but leaves less power for computation at the transmitter, correspondingly increasing its latency.
Furthermore, inter-satellite WPT incurs losses from electrical--optical conversion and beam spreading over the air, which complicates power control.
Therefore, effective WPT control requires optimizing transfer directions and power levels according to the computation demands at both ends and the associated link efficiencies.

To address these challenges, we present SpaceWPT, a WPT-assisted inference framework for low-latency token generation in SpaceDCs.
SpaceWPT builds on the existing SpaceMoE architecture, where the attention and gating modules are deployed on a gateway satellite, while the expert modules are distributed across the remaining satellites~\cite{wang2026spacemoe}.
Each token is processed by the attention module at the gateway and then by the activated experts in parallel, whose outputs are aggregated at the gateway.
To the best of our knowledge, this is the first attempt to deploy WPT for LLM inference in space.
The key contributions and findings of this work are summarized as follows.

\begin{itemize}
    \item \textbf{Attention--Expert Latency Tradeoff:}
    We theoretically characterize the effect of WPT power control on E2E MoE inference latency.
    To this end, two basic design principles are derived to guide the control of transfer directions and power levels.
    First, it is optimal for a satellite to transfer all its harvested power when its expert is inactive. 
    Second, it never transmits and receives power simultaneously when its expert is active.
    Based on these principles, we construct a tractable surrogate for the inference latency under stochastic expert activation.
    The resulting surrogate reveals an attention--expert latency tradeoff: power transferred by an active expert to support attention computation reduces attention latency while increasing the expert latency.
    Balancing this tradeoff to reduce the total inference latency necessitates the optimization of power control in the sequel.

    \item \textbf{Activation-Based WPT Control:}
    To balance the tradeoff between attention and expert latencies, we design a near-optimal strategy for WPT power controls by minimizing the latency surrogate.
    The resulting solution exhibits a \emph{two-sided water-filling} structure.
    Compared with conventional water filling, which allocates power from a single source to multiple links, this structure allows each satellite to donate, retain, or receive power.
    In this structure, the water level is determined by expert activation statistics, the MoE inference workload, and WPT link efficiencies.
    Then, the transfer direction and power level of each satellite can be obtained by comparing its link efficiency with the water level.
    The solution shows that satellites with more efficient WPT links have higher power levels, while those with inefficient links tend to retain their harvested power.
    This strategy reduces the original high-dimensional problem to a single scalar equation for the water level while providing a provable near-optimality guarantee.

    \item \textbf{Experiments:}
    Experiments with Mixtral-8x7B and Phi-3.5-MoE on HumanEval confirm the derived attention--expert tradeoff and the substantial latency reductions of SpaceWPT over No-WPT baselines.
    The latency gain increases as WPT efficiency improves and the SpaceDC formation becomes more compact, highlighting the benefit of efficient short-range WPT.
\end{itemize}

The remainder of this paper is organized as follows.
Section~\ref{sec:system_model} presents the system model and performance metrics.
Section~\ref{sec:framework} analyzes the inference latency and derives a tractable surrogate.
Section~\ref{sec:power_control} develops the WPT power-control strategy, and Section~\ref{sec:experiments} presents the experimental results.

\section{System Model and Metrics}
\label{sec:system_model}

\subsection{MoE Inference Model}
\label{sec:computation_model}

\subsubsection{MoE Architecture and Inference}
\label{sec:moe_architecture}

We consider a decoder-based large language model whose $L$ decoder layers employ the MoE architecture, illustrated in Fig.~\ref{fig:moe_architecture}.
Each decoder layer comprises three modules: a multi-head attention module shared by every token, a gating module that selects the active experts, and an expert module containing $K$ expert \emph{feed-forward networks} (FFNs).
First, each token's input is processed by the attention module, which requires $\Catt$ FLOPs.
Second, a gating matrix $\bm{W}_{\mathsf{g}}$ maps the output of the attention module to a score vector.
Third, $\kappa$ experts with the highest gating scores are activated, forming the active expert set $\mathcal{A}(\mathbf{s})\triangleq\{k:s_k=1\}$, where $s_k$ is the binary activation indicator for expert~$k$.
The activated experts process the token in parallel, each requiring $\Cffn$ FLOPs.
The gateway then aggregates their outputs using the normalized gating weights to obtain the output of the current MoE layer.

\subsubsection{Expert Activation Model}

To quantify the activation distribution, we adopt a classical \emph{probability-proportional-to-size without replacement} (PPSWOR) model for sampling $\kappa$ experts out of $K$ candidates~\cite{ppswor_rosen}.
This theory associates each expert~$k$ with an importance weight $\omega_k>0$.
The resulting  \emph{probability mass function} (PMF) of top-$\kappa$ experts is given by
\begin{equation}
    \Pr\{\mathcal{A}(\mathbf{s})=\mathcal{U}\}=\frac{\prod_{k\in\mathcal{U}}\omega_k}{e_\kappa(\omega_1,\dots,\omega_K)},\qquad |\mathcal{U}|=\kappa,
    \label{eq:ppswor_pmf}
\end{equation}
where $\mathcal{U}\subseteq\{1,\dots,K\}$ and $e_\kappa$ is the $\kappa$-th elementary symmetric polynomial,
\begin{equation*}
    e_\kappa(\omega_1,\dots,\omega_K)\triangleq\sum_{\substack{\mathcal{U}\subseteq\{1,\dots,K\}\\ |\mathcal{U}|=\kappa}}\prod_{k\in\mathcal{U}}\omega_k.
\end{equation*}
Note that the per-expert \emph{activation probability} of this model is $\alpha_k\triangleq\Pr\{k\in\mathcal{A}(\mathbf{s})\}=\E[s_k]$, which is a monotone increasing function of $\omega_k$ and can be computed by
\begin{equation}
    \alpha_k=1-\frac{e_\kappa(\omega_1,\dots,\hat{\omega}_k,\dots,\omega_K)}{e_\kappa(\omega_1,\dots,\omega_K)},
    \label{eq:alpha_k_def}
\end{equation}
where $\hat{\omega}_k$ denotes omission of $\omega_k$.

\begin{figure}[t]
    \centering
    \includegraphics[width=1.02\columnwidth]{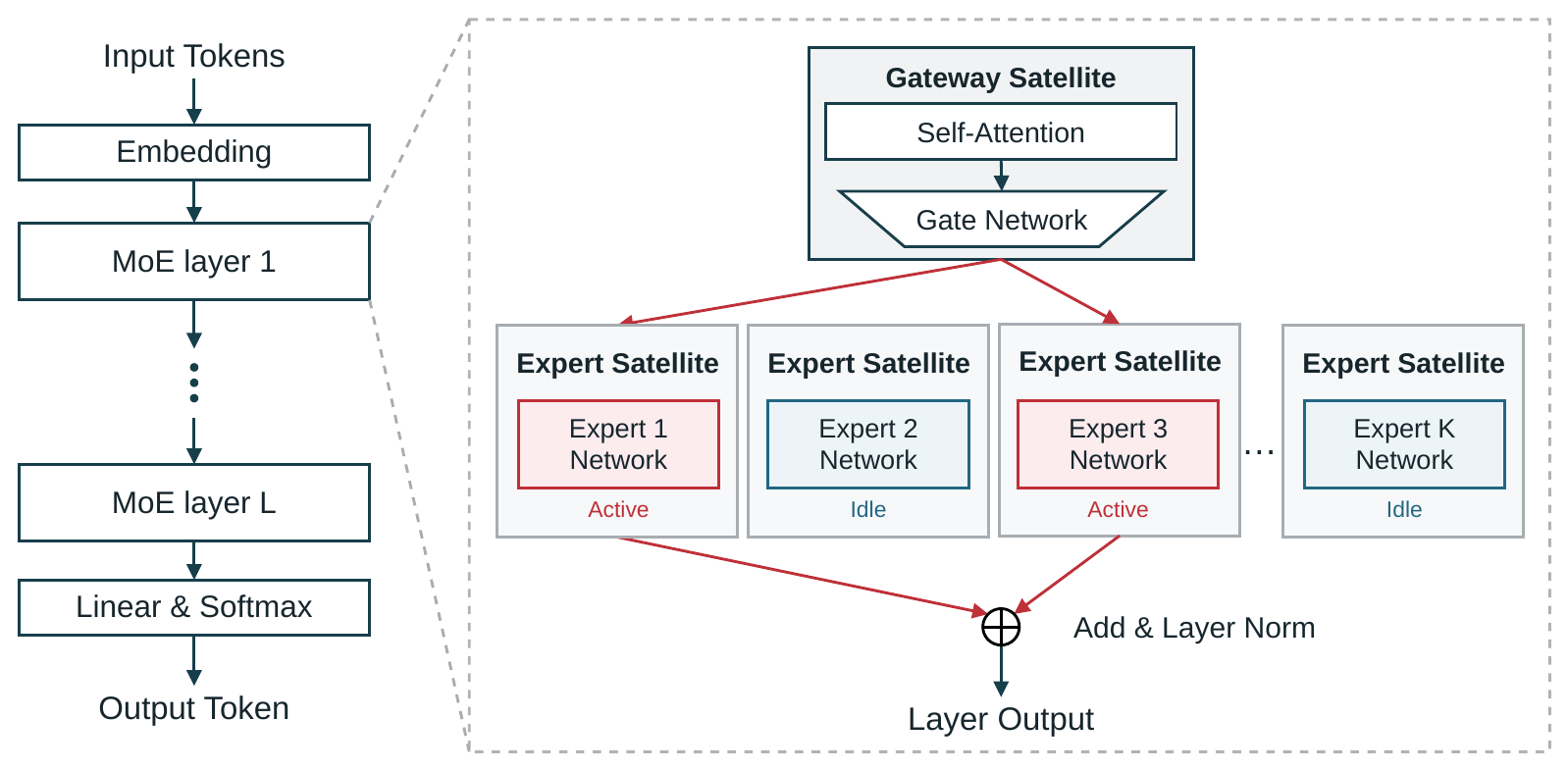}
    \caption{MoE inference architecture and its deployment.}
    \label{fig:moe_architecture}
\end{figure}

\subsubsection{Computation Latency Model}

We model the computation latency at each processor under \emph{dynamic voltage and frequency scaling} (DVFS).
The relation between the computation speed and the allocated compute power $P^{\sf{cmp}}$ is modeled as~\cite{dvfs_chandrakasan}
\begin{equation}
    f = \left(\frac{P^{\sf{cmp}}}{\gamma}\right)^{\!1/3},
    \label{eq:dvfs}
\end{equation}
where $f$ is measured in FLOPs per second and $\gamma>0$ is an architecture-dependent coefficient.
The resulting latency of executing a workload of $C$ FLOPs on such a processor is given by
\begin{equation}
    \Tcmp = C\left(\frac{\gamma}{P^{\sf{cmp}}}\right)^{\!1/3}.
    \label{eq:latency_generic}
\end{equation}

\subsection{Modelling Space Data Center}
\label{sec:sdc_model}

\subsubsection{Expert-to-Satellite Mapping}
\label{sec:sdc_mapping}

We consider a SpaceDC comprising a formation of $K+1$ satellites that jointly execute the MoE model.
For each MoE layer, the attention and gating modules are deployed on one satellite, designated as the \emph{gateway satellite} and indexed by~$0$.
Each of the remaining $K$ satellites hosts one expert FFN and is referred to as an \emph{expert satellite}.
Under this deployment, each expert and its host satellite share the index $k\in\Kset\triangleq\{1,\ldots,K\}$.

\subsubsection{Geometry of Space Data Center}
\label{sec:sdc_orbit}

Following the orbital design recently proposed in~\cite{y2025towards}, the SpaceDC illustrated in Fig.~\ref{fig:sdc_orbit_overview} operates as a formation-flying cluster in a dawn--dusk SSO.
Since this orbit approximately follows Earth's day--night terminator, it provides near-constant solar illumination~\cite{virgili2015descending}.
We model each satellite's harvested solar power as a constant budget $\Ps$.

To describe the cluster geometry, we adopt a \emph{local-vertical-local-horizontal} (LVLH) coordinate system that co-rotates with the gateway satellite, as illustrated in Fig.~\ref{fig:sdc_orbit_overview}: the $x$-axis points radially outward from the Earth, and the $y$-axis aligns with the gateway satellite's instantaneous velocity.
In this coordinate system, let $(x_{k}(t),y_{k}(t))$ denote the position of expert satellite~$k$ relative to the gateway satellite.
The relative motion is governed by the \emph{Hill-Clohessy-Wiltshire} (HCW) equations~\cite{clohessy1960terminal}
\begin{equation}
    \ddot{x}_{k} - 2\omega\dot{y}_{k} - 3\omega^2 x_{k} = 0,\quad
    \ddot{y}_{k} + 2\omega\dot{x}_{k} = 0,
    \label{eq:hcw}
\end{equation}
where $\omega$ denotes the orbital angular speed of the gateway satellite.
When the initial radial offset and along-track speed satisfy $\dot{y}_{k}(0)=-2\omega\,x_{k}(0)$, expert satellite~$k$ stays on a closed relative orbit around the gateway instead of drifting away.
The solutions of~\eqref{eq:hcw} admit the closed form, given by
\begin{equation}
    \begin{aligned}
        x_{k}(t) &= \rho_{k}\cos(\omega t+\phi_{k}),\\
        y_{k}(t) &= -2\rho_{k}\sin(\omega t+\phi_{k}) + y_{k}^{\mathsf{off}},
    \end{aligned}
    \label{eq:hcw_sol}
\end{equation}
where $\rho_{k}$ sets the size of the relative orbit, $\phi_{k}$ sets the starting position on this orbit, and $y_{k}^{\mathsf{off}}$ is the satellite's average along-track offset from the gateway.
Geometrically,~\eqref{eq:hcw_sol} traces a closed 2:1 ellipse centered on $(0,\,y_{k}^{\mathsf{off}})$ with period $T_{\mathsf{orb}}=2\pi/\omega$.
Thus, expert satellite~$k$'s distance from the gateway satellite is
\begin{equation}
    d_{k}(t) = \sqrt{x_{k}^2(t) + y_{k}^2(t)},
    \label{eq:dk_t}
\end{equation}
which varies on the orbital timescale $T_{\mathsf{orb}}\sim 10^{3}$~s.
In contrast, each token inference latency is on the order of milliseconds, six orders of magnitude shorter than $T_{\mathsf{orb}}$~\cite{pope2023efficiently}.
Therefore, the cluster geometry is approximately fixed within any single inference pass, and we treat $\{d_{k}(t)\}$ as constants $\{d_{k}\}$ in the remainder of this paper.
The resulting gateway-to-expert distances $\{d_{k}\}$ determine the link-specific communication and WPT characteristics developed in the following two subsections.

\begin{figure}[t]
    \centering
    \includegraphics[width=0.95\columnwidth]{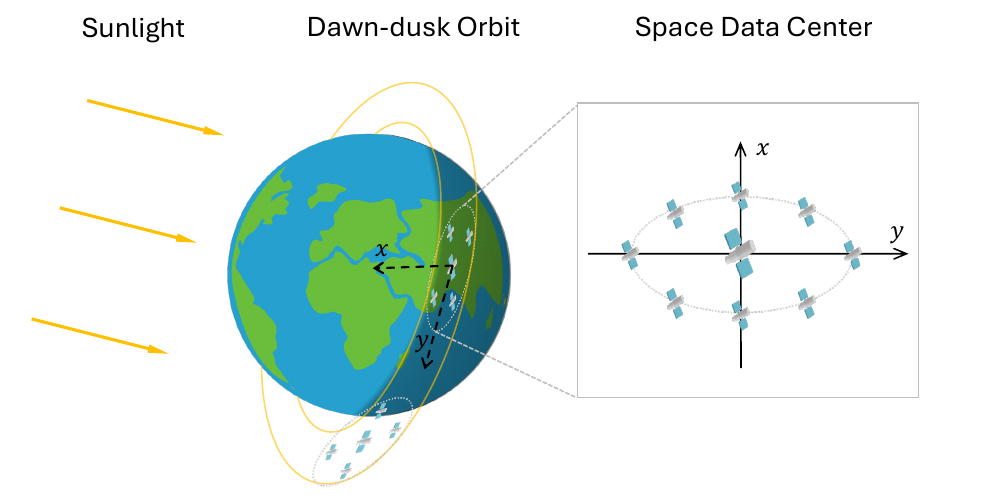}
    \caption{SpaceDCs in a dawn–dusk orbit and the relative-orbit geometry of their satellite clusters.}
    \label{fig:sdc_orbit_overview}
\end{figure}

\begin{figure*}[t]
    \centering
    \makebox[\textwidth][c]{%
    \subfigure[SpaceWPT system.]{%
        \includegraphics[height=0.230\textwidth,trim=5 8 3 4,clip]{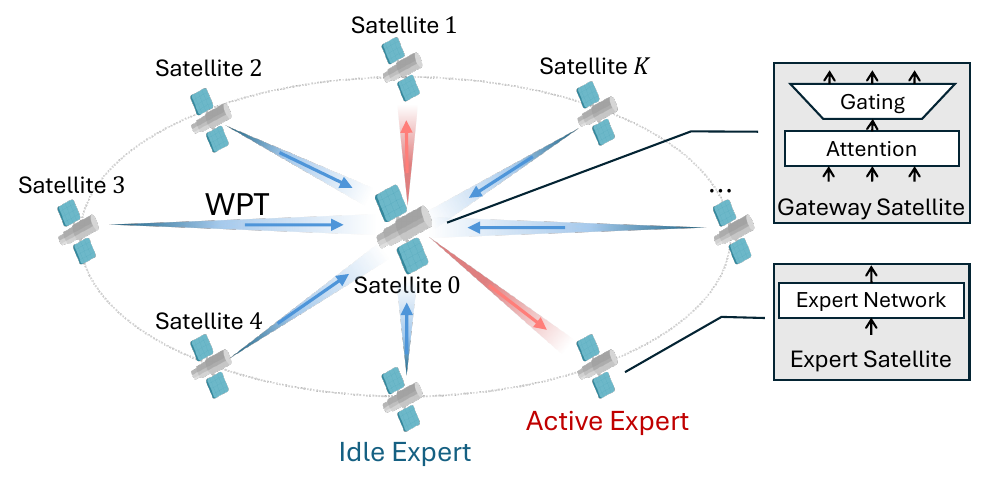}%
        \label{fig:space_wpt_system}%
    }%
    \hspace{0.02\textwidth}%
    \subfigure[System operations.]{%
        \includegraphics[height=0.230\textwidth,trim=64 116 18 8,clip]{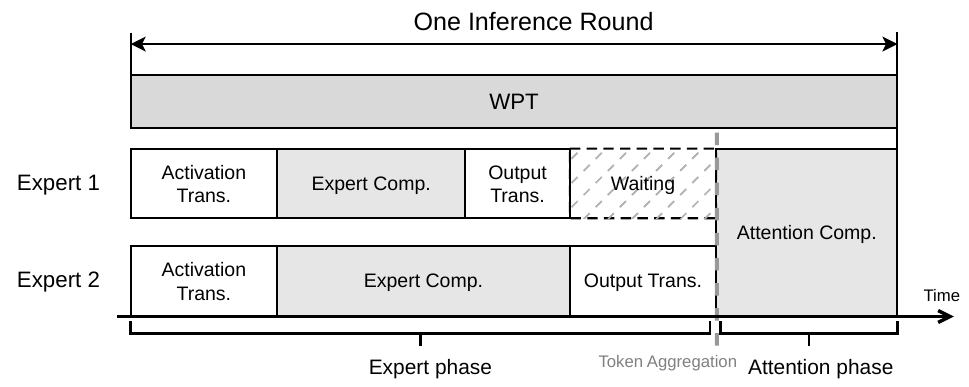}%
        \label{fig:system_operation}%
    }}%
    \centeredfigurecaption{SpaceWPT system and operations.}%
      {fig:spacewpt_system_operations}
\end{figure*}

\subsection{Communication Model}
\label{sec:communication_model}

Inter-satellite data exchange employs \emph{free-space optical} (FSO) links operating in a diffraction-limited regime~\cite{majumdar2008fso,kaushal2017fso}. 
For a link with transmit and receive aperture diameters $D_{\mathsf{tx}}^{\mathsf{c}}$ and $D_{\mathsf{rx}}^{\mathsf{c}}$ at wavelength $\lambda^{\mathsf{c}}$, the channel power gain between two nodes separated by distance $d$ is
\begin{equation}
    h(d) = \eta^{\mathsf{c}}
      \!\left(\frac{\pi\, D_{\mathsf{tx}}^{\mathsf{c}} D_{\mathsf{rx}}^{\mathsf{c}}}
                    {4\,\lambda^{\mathsf{c}}\,d}\right)^{\!2},
    \label{eq:fso_gain}
\end{equation}
where $\eta^{\mathsf{c}}\in(0,1]$ absorbs pointing errors and optical-path losses.
Applying this model to gateway--expert link~$k$ at distance $d_k$ gives
\begin{equation*}
    h_k \triangleq h(d_k).
\end{equation*}
Hence, different gateway-to-expert distances result in heterogeneous communication-channel gains across the SpaceDC.
With transmit power $P^{\mathsf{cmm}}$, noise spectral density $N_0$, and link bandwidth $B$, the achievable data rate over the FSO link is~\cite{goldsmith2005wireless}
\begin{equation}
    R(d) = B\,\log_2\!\left(1+\frac{P^{\mathsf{cmm}}\,h(d)}{N_0 B}\right).
    \label{eq:fso_rate}
\end{equation}
Transmitting $Q$ bits over a link of distance $d$ takes a one-way latency
\begin{equation}
    \Tcmm(d, Q) = \frac{Q}{R(d)} + \frac{d}{c},
    \label{eq:one_way_tcmm}
\end{equation}
where $c$ is the speed of light and the second term is the propagation delay.

\subsection{Modelling Wireless-Power-Transfer}
\label{sec:wpt}

\subsubsection{Laser-Based WPT Model}

We adopt \emph{laser-based WPT}~\cite{PowerHotspots,mohsan2023owpt} for inter-satellite power transmission in the cluster.
Each satellite carries a dedicated high-power laser as the WPT transmitter, together with an on-board \emph{photovoltaic} (PV) receiver that converts incident laser irradiance into electrical power~\cite{fafard2021photovoltaic,helmers2021photonic}.
Any satellite can therefore act as a WPT source for any other, over an optical channel separate from the data link.

Following the same diffraction-limited aperture-coupling law as in~\eqref{eq:fso_gain}, the optical channel power gain at link distance $d$ is
\begin{equation}
    g(d) = \eta^{\mathsf{w}}
      \!\left(\frac{\pi\, D_{\mathsf{tx}}^{\mathsf{w}} D_{\mathsf{rx}}^{\mathsf{w}}}
                    {4\,\lambda^{\mathsf{w}}\,d}\right)^{\!2},
    \label{eq:wpt_gain}
\end{equation}
where $\lambda^{\mathsf{w}}$, $D_{\mathsf{tx}}^{\mathsf{w}}$, and $D_{\mathsf{rx}}^{\mathsf{w}}$ are the laser wavelength and transmit/receive aperture diameters, and $\eta^{\mathsf{w}}\in(0,1]$ absorbs optical path losses.
WPT energy traverses three cascaded conversion stages: an \emph{electrical-to-optical} (E/O) stage at the transmit laser with wall-plug efficiency $\eta_{\mathsf{eo}}\in(0,1]$, the optical channel in~\eqref{eq:wpt_gain}, and an \emph{optical-to-electrical} (O/E) stage at the PV receiver with conversion efficiency $\eta_{\mathsf{oe}}\in(0,1]$.
The cascade yields the end-to-end electrical-to-electrical transfer efficiency
\begin{equation}
    \eta(d) \triangleq \eta_{\mathsf{eo}}\, g(d)\, \eta_{\mathsf{oe}}.
    \label{eq:eta_end_to_end}
\end{equation}
For gateway--expert link~$k$ at distance $d_k$, the corresponding WPT efficiency is
\begin{equation}
    \eta_k \triangleq \eta(d_k)\in(0,1).
    \label{eq:etak_def}
\end{equation}
Hence, different gateway-to-expert distances result in heterogeneous WPT efficiencies across the SpaceDC.
Let $P$ denote the electrical input power to the transmit laser.
Accordingly, the harvested electrical power at the receiver is then $\eta(d)\,P$.
When this link is directed from node~$i$ to node~$j$, we write the transmit power as $P_{i\to j}$, the convention used throughout the remainder of the paper.

\subsubsection{WPT Topology}
\label{sec:sdc_wpt_topology}

We consider the \emph{star} WPT topology illustrated in Fig.~\ref{fig:space_wpt_system}, in which each expert exchanges power only with the gateway.
At every node, the available compute power equals the local solar-power budget plus the received WPT power minus the transmitted WPT power.
The gateway aggregates contributions from all $K$ experts, given by
\begin{equation}
    P_0^{\mathsf{cmp}} = \Ps + \sum_{k=1}^{K}\eta_k\,P_{k\to 0} - \sum_{k=1}^{K}P_{0\to k},
    \label{eq:P0_cmp_generic}
\end{equation}
while each expert~$k$ couples only with the gateway, given by
\begin{equation}
    P_k^{\mathsf{cmp}} = \Ps + \eta_k\,P_{0\to k} - P_{k\to 0}.
    \label{eq:Pk_cmp_generic}
\end{equation}
Each WPT transmit power is non-negative and capped by the local solar-power budget, given by
\begin{equation}
    0\le P_{k\to 0} \le \Ps,
    \qquad
    0\le P_{0\to k} \le \Ps.
    \label{eq:per_node_tx_caps}
\end{equation}

\subsection{System Operations}
\label{sec:operation_flow}

We adopt expert parallelism for distributed MoE inference in the SpaceDC~\cite{rajbhandari2022deepspeedmoe}.
Each round begins with a gating decision at the gateway and ends before the subsequent gating decision.
The operations within one round are illustrated in Fig.~\ref{fig:system_operation} and described as follows.
\begin{enumerate}
    \item \emph{Gating:}
    Given the attention output from the preceding round, the gateway performs the gating operation described in Section~\ref{sec:computation_model}, yielding the activation state~$\mathbf{s}$ and the active expert set~$\mathcal{A}(\mathbf{s})$ defined in Section~\ref{sec:moe_architecture}.

    \item \emph{WPT Power Control:}
    Based on~$\mathbf{s}$, the gateway configures the WPT transmit powers using the power-control solution derived in Section~\ref{sec:power_control}, over the star topology described in Section~\ref{sec:sdc_wpt_topology}.

    \item \emph{Activation Transmission:}
    The gateway transmits the input activation to the active experts over the FSO data links modeled in Section~\ref{sec:communication_model}.

    \item \emph{Expert Computation:}
    Each active expert processes the received activation using its onboard FFN, following the MoE inference model in Section~\ref{sec:computation_model}.

    \item \emph{Token Aggregation:}
    Each active expert returns its FFN output to the gateway over the FSO data link.
    The gateway aggregates the received outputs to obtain the layer output.

    \item \emph{Attention Computation:}
    The gateway computes the attention output from the aggregated expert output and passes it to the gating module, thereby initiating the next round.
\end{enumerate}
During Steps~3--5, the active-expert paths operate in parallel, where each expert starts its FFN computation upon receiving the activation and returns its output upon completion.
The WPT powers configured in Step~2 are held fixed throughout Steps~3--6. Separate optical links allow WPT to operate concurrently with data communication.

\subsection{Performance Metrics}
\label{sec:metrics}
 
\subsubsection{Inference Accuracy}
\label{sec:metric_accuracy}
 
The inference accuracy of the MoE model depends on its gating matrix and expert FFNs, both of which remain fixed after deployment.
The WPT power allocation changes only the compute power available at each node.
Inference accuracy is therefore independent of the WPT power allocation.
 
\subsubsection{Inference Latency}
\label{sec:metric_latency}

We therefore adopt the per-token inference latency as the performance metric.
The gating and WPT-control signaling in Steps~1--2 incur negligible latency and are omitted.

The inference latency therefore comprises the parallel expert phase in Steps~3--5 and the gateway-side attention computation in Step~6.
Let $Q^{\mathsf{dl}}$ and $Q^{\mathsf{ul}}$ denote the input-activation and expert-output sizes, respectively.
Following the FSO communication model in Section~\ref{sec:communication_model}, the communication latency along expert~$k$'s path is
\begin{equation}
    T_k^{\mathsf{cmm}} \triangleq \Tcmm(d_k,Q^{\mathsf{dl}}) + \Tcmm(d_k,Q^{\mathsf{ul}}).
    \label{eq:Tkcmm_def}
\end{equation}
Combining this communication latency with the computation model in Section~\ref{sec:computation_model} and the parallel execution of Steps~3--5 gives the expert-phase latency
\begin{equation}
    T^{\mathsf{exp}}\!\left(\left\{P_k^{\mathsf{cmp}}\right\};\mathbf{s}\right)
    = \max_{1\le k\le K}s_k\left[T_k^{\mathsf{cmm}}+\Cffn\left(\frac{\gamma}{P_k^{\mathsf{cmp}}}\right)^{\!\frac{1}{3}}\right],
    \label{eq:expert_bottleneck_qualitative}
\end{equation}
where $P_k^{\mathsf{cmp}}$ is the compute power available at expert~$k$, as given in~\eqref{eq:Pk_cmp_generic}.
After token aggregation, the gateway performs the attention computation in Step~6.
Following the computation model in Section~\ref{sec:computation_model}, its latency is
\begin{equation}
    \Tatt\!\left(P_0^{\mathsf{cmp}}\right) = \Catt\left(\frac{\gamma}{P_0^{\mathsf{cmp}}}\right)^{\!1/3},
    \label{eq:center_latency_term}
\end{equation}
where $P_0^{\mathsf{cmp}}$ is the compute power available at the gateway, as given in~\eqref{eq:P0_cmp_generic}.
The gateway and expert compute powers are collected in the vector $\mathbf{P}^{\mathsf{cmp}}\triangleq\left(P_0^{\mathsf{cmp}},P_1^{\mathsf{cmp}},\ldots,P_K^{\mathsf{cmp}}\right)$.
The per-token inference latency is given by
\begin{equation}
    T\!\left(\mathbf{P}^{\mathsf{cmp}};\mathbf{s}\right)
    = T^{\mathsf{exp}}\!\left(\left\{P_k^{\mathsf{cmp}}\right\};\mathbf{s}\right)
    + \Tatt\!\left(P_0^{\mathsf{cmp}}\right),
    \label{eq:per_token_latency_physical}
\end{equation}
The optimization objective is the average per-token latency over the activation model in~\eqref{eq:ppswor_pmf}, given by
\begin{equation}
    \bar T (\mathbf{P}^{\mathsf{cmp}}) \triangleq \E_{\mathbf{s}}\!\left[T\!\left(\mathbf{P}^{\mathsf{cmp}};\mathbf{s}\right)\right].
    \label{eq:avg_latency_physical_def}
\end{equation}

\section{Inference Latency Analysis}
\label{sec:framework}

In this section, we formulate the inference latency through the lens of the WPT power control and then derive a tractable surrogate for its average.
The surrogate reveals the power--latency tradeoff between the gateway and the experts, which is exploited in the next section to optimize the WPT control.

\subsection{Latency-Minimizing Problem}
\label{sec:latency_formulation}

In this subsection, we analyze the per-token inference latency as a function of the WPT power and formulate the latency-minimizing problem.
The compute powers in~\eqref{eq:per_token_latency_physical} are determined by the local solar-power budget and WPT transfers through~\eqref{eq:P0_cmp_generic} and~\eqref{eq:Pk_cmp_generic}.
However, each expert power affects the latency only through the expert latency term~\eqref{eq:expert_bottleneck_qualitative}, so increasing it reduces the latency only when expert~$k$ is active.
Since the relation between the powers and the latency depends on the activation state, we allow each expert to set its WPT powers separately for its active and idle states, giving the state-dependent functions $P_{k\to0}(s_k)$ and $P_{0\to k}(s_k)$ for $s_k\in\{0,1\}$.
Thus, each expert carries four optimization variables, a donating and a receiving power in each of the two activation states, which are optimized to minimize the average latency.

We then formulate the latency-minimizing problem over the state-dependent WPT powers.
Stacking them across all experts and activation states gives the optimization vector
\begin{equation}
    \mathbf{q}\triangleq\bigl\{P_{k\to0}(s_k),\,P_{0\to k}(s_k)\bigr\}_{k=1,\ldots,K;\;s_k\in\{0,1\}},
    \label{eq:raw_decision_vector}
\end{equation}
which induces $\mathbf{P}^{\mathsf{cmp}}(\mathbf{q};\mathbf{s})$ through~\eqref{eq:P0_cmp_generic} and~\eqref{eq:Pk_cmp_generic} for each activation state $\mathbf{s}$.
Accordingly, the per-token latency induced by $\mathbf{q}$ is
\begin{equation*}
    T(\mathbf{q};\mathbf{s}) \triangleq
    T\!\left(\mathbf{P}^{\mathsf{cmp}}(\mathbf{q};\mathbf{s});\mathbf{s}\right).
\end{equation*}
Substituting this induced latency into~\eqref{eq:avg_latency_physical_def} gives
\begin{equation*}
    \bar T(\mathbf{q}) \triangleq \E_{\mathbf{s}}\!\left[T(\mathbf{q};\mathbf{s})\right].
\end{equation*}
Thus, the latency-minimizing problem is
\begin{align}
    (\mathcal{P}_0):\quad\min_{\mathbf{q}}\quad & \bar T(\mathbf{q}) \nonumber\\
    \mathrm{s.t.}\quad & 0\le P_{k\to0}(s_k),\,P_{0\to k}(s_k)\le\Ps, \nonumber\\
    & P_j^{\mathsf{cmp}}>0,\quad j\in\{0\}\cup\mathcal A(\mathbf s), \nonumber\\
    & k=1,\ldots,K,\ s_k\in\{0,1\}. \nonumber
\end{align}
Problem \((\mathcal P_0)\) is intractable for two reasons.
First, the optimization involves $4K$ variables in~\eqref{eq:raw_decision_vector}.
Second, the expectation of the per-token latency over the random activation state $\mathbf{s}$ has no closed-form expression.
We address the first challenge by exploiting the two design principles established in Section~\ref{sec:structural_properties} to simplify the problem in Section~\ref{sec:problem_simplification}, and the second by deriving a tractable surrogate in Section~\ref{sec:surrogate}.

\subsection{Design Principles for Optimal WPT Control}
\label{sec:structural_properties}

In this subsection, we establish two design principles for optimal WPT control that reduce the number of optimization variables in problem \((\mathcal{P}_0)\).

First, consider an idle-state expert, where it runs no FFN, so its WPT power can serve only the gateway.
The optimal solution is then to donate its full solar-power budget to the gateway, as established in Lemma~\ref{lem:bi_idle_structure}.
\begin{lemma}[Optimal idle-state powers]
\label{lem:bi_idle_structure}
Every optimal power control satisfies
\begin{equation}
    P_{k\to0}^{\star}(s_k{=}0)=\Ps,\quad P_{0\to k}^{\star}(s_k{=}0)=0,\quad k=1,\ldots,K.
    \label{eq:bi_idle_optimal}
\end{equation}
\end{lemma}
\begin{proof}
An idle expert runs no FFN, so its powers affect only the gateway power $P_0^{\mathsf{cmp}}$ and leave the expert-side latency unchanged.
Since $P_0^{\mathsf{cmp}}$ increases in $P_{k\to0}(s_k{=}0)$ and decreases in $P_{0\to k}(s_k{=}0)$, while a larger $P_0^{\mathsf{cmp}}$ always lowers the gateway latency, the optimum sets $P_{k\to0}(s_k{=}0)=\Ps$ and $P_{0\to k}(s_k{=}0)=0$.
\end{proof}
Second, consider an active-state expert, where it may both send power to and receive power from the gateway.
Transmitting in both directions at once wastes power on the round trip, so an optimal expert keeps only one direction and sets the other to zero, as established in Lemma~\ref{lem:no_counter_flow}.
\begin{lemma}[No simultaneous counter-flow]
\label{lem:no_counter_flow}
Every optimal power control satisfies
\begin{equation}
    P_{k\to0}^{\star}(s_k{=}1)\,\cdot P_{0\to k}^{\star}(s_k{=}1)=0,\qquad k=1,\ldots,K.
    \label{eq:no_counter_flow}
\end{equation}
\end{lemma}
\begin{proof}
Suppose an optimal control has $P_{k\to0}(s_k{=}1)>0$ and $P_{0\to k}(s_k{=}1)>0$ for some expert~$k$.
Reducing the two powers to $P_{k\to0}(s_k{=}1)-\eta_k\varepsilon$ and $P_{0\to k}(s_k{=}1)-\varepsilon$ with $\varepsilon=\min\{P_{0\to k}(s_k{=}1),\,P_{k\to0}(s_k{=}1)/\eta_k\}$ drives one direction to zero while leaving the expert compute power unchanged.
This cuts the power the gateway sends out by $\varepsilon$ but the power it collects from the expert by only $\eta_k^2\varepsilon$, so $P_0^{\mathsf{cmp}}$ gains $\varepsilon(1-\eta_k^2)>0$. 
Hence, the per-token latency strictly decreases, contradicting optimality.
\end{proof}

The two design principles simplify WPT control by fixing idle-state donation and eliminating simultaneous transmission and reception in the active state.
Since Lemma~\ref{lem:no_counter_flow} precludes simultaneous counter-flow, the active-state power exchange can be represented without loss of optimality by a single signed scalar~\cite{ahuja1993networkflows}
\begin{equation}
    P_k \triangleq P_{k\to0}(s_k{=}1)-\eta_k P_{0\to k}(s_k{=}1),
    \quad k=1,\ldots,K,
\label{eq:bi_pk_def}
\end{equation}
The cases $P_k>0$, $P_k=0$, and $P_k<0$ correspond to donation,
neutrality, and reception, respectively. In the reception case,
$|P_k|=\eta_kP_{0\to k}(s_k{=}1)$ is the power delivered to expert~$k$
after transfer loss.

\subsection{Inference Latency under Simplified WPT Control}
\label{sec:problem_simplification}

With WPT control simplified due to the preceding design principles, we reformulate the inference latency and problem \((\mathcal{P}_0)\) in terms of the signed-power vector $\mathbf{p}=(P_1,\ldots,P_K)$.

Firstly, we explore the relation between $P_k$ and the expert compute power, which affects the expert-side latency in~\eqref{eq:expert_bottleneck_qualitative}.
Since idle experts run no FFN, the expert-side latency depends only on active experts.
For active states $s_k=1$, the compute power of expert~$k$ is obtained from~\eqref{eq:Pk_cmp_generic} as
\begin{equation}
    P_k^{\mathsf{cmp}}(s_k{=}1) = \Ps + \eta_k P_{0\to k}(s_k{=}1) - P_{k\to 0}(s_k{=}1),
    \label{eq:Pk_cmp_active}
\end{equation}
where $P_{k\to0}(s_k{=}1)$ and $P_{0\to k}(s_k{=}1)$ can be expressed in terms of $P_k$ as
\begin{equation}
    P_{k\to0}(s_k{=}1) = (P_k)^{+},\quad
    P_{0\to k}(s_k{=}1) = \frac{(-P_k)^{+}}{\eta_k},
    \label{eq:bi_active_from_pk}
\end{equation}
with $(z)^{+}\triangleq\max\{z,0\}$.
Substituting~\eqref{eq:bi_active_from_pk} into~\eqref{eq:Pk_cmp_active} expresses this compute power through $P_k$ as
\begin{equation}
    P_k^{\mathsf{cmp}}(s_k{=}1) = \Ps - P_k,\qquad k=1,\ldots,K,
    \label{eq:Pk_cmp_from_wpt}
\end{equation}
Thus, the expert-phase latency~\eqref{eq:expert_bottleneck_qualitative} can be expressed through $\mathbf{p}$, given by
\begin{equation}
    T^{\mathsf{exp}}(\mathbf{p};\mathbf{s})=\max_{1\le k\le K}\left\{s_k\left[T_k^{\mathsf{cmm}}+\lffn\,(\Ps-P_k)^{-\frac{1}{3}}\right]\right\},
    \label{eq:expert_latency_pk}
\end{equation}
where $\lffn\triangleq\Cffn\gamma^{1/3}$ sets an expert's FFN compute latency at a given compute power.

Secondly, we explore the relation between $\mathbf{p}$ and the gateway compute power, which affects the gateway-side latency in~\eqref{eq:center_latency_term}.
Unlike the expert side, the gateway exchanges power with every expert in both activation states, so its compute power aggregates over all experts, expressed as
\begin{equation}
    P_0^{\mathsf{cmp}}(\mathbf{s}) = \Ps + \sum_{k=1}^{K}\bigl[\eta_k\,P_{k\to 0}(s_k) - P_{0\to k}(s_k)\bigr],
    \label{eq:P0_cmp_generic_repeated}
\end{equation}
where these transmit powers combine the active-state transfer~\eqref{eq:bi_active_from_pk} with the idle-state donation $\Ps$ of Lemma~\ref{lem:bi_idle_structure}, given by
\begin{equation}
\begin{aligned}
    P_{k\to0}(s_k) &= s_k\,(P_k)^{+}+(1-s_k)\,\Ps,\\
    P_{0\to k}(s_k) &= s_k\,\frac{(-P_k)^{+}}{\eta_k}.
\end{aligned}
    \label{eq:instantaneous_wpt_combined}
\end{equation}
Substituting~\eqref{eq:instantaneous_wpt_combined} into the gateway compute power~\eqref{eq:P0_cmp_generic_repeated}, the contribution of an active expert depends on $P_k$ and is given by
\begin{equation*}
    \eta_k P_{k\to0}(s_k{=}1)-P_{0\to k}(s_k{=}1)=\eta_k(P_k)^{+}-\frac{(-P_k)^{+}}{\eta_k}.
\end{equation*}
Since $(P_k)^{+}$ and $(-P_k)^{+}$ cannot both be nonzero, the expression reduces according to the sign of $P_k$ to the piecewise function $g_k(P_k)$, written as
\begin{equation}
    g_k(P_k)\triangleq\eta_k(P_k)^{+}-\frac{(-P_k)^{+}}{\eta_k}=
    \begin{cases}
      \eta_k P_k, & P_k\ge0,\\[0.4ex]
      P_k/\eta_k, & P_k<0,
    \end{cases}
    \label{eq:bi_ellk}
\end{equation}
so a donating expert ($P_k\ge0$) contributes its harvested power $\eta_k P_k$ to the gateway, whereas a receiving one ($P_k<0$) draws $|P_k|/\eta_k$ from the gateway.
Summing over donated and received power from all experts yields the gateway compute power through $\mathbf{p}$, given by
\begin{equation}
    P_0^{\mathsf{cmp}}(\mathbf{p}; \mathbf{s})=\Ps+\sum_{k=1}^{K}\bigl[(1-s_k)\eta_k\Ps+s_k\,g_k(P_k)\bigr],
    \label{eq:P0_cmp_instantaneous}
\end{equation}
Substituting~\eqref{eq:P0_cmp_instantaneous} into the attention latency~\eqref{eq:center_latency_term} then expresses the gateway-side latency through $\mathbf{p}$, given by
\begin{equation}
\Tatt(\mathbf{p}; \mathbf{s})=\latt\,\left(\Ps+\sum_{k=1}^{K}\bigl[(1-s_k)\eta_k\Ps+s_k\,g_k(P_k)\bigr]\right)^{-\frac{1}{3}},
    \label{eq:att_latency_pk}
\end{equation}
where $\latt\triangleq\Catt\gamma^{1/3}$ sets the attention compute latency at a given gateway compute power.

Combining the expert-side and gateway-side latencies, the per-token inference latency~\eqref{eq:per_token_latency_physical} becomes
\begin{equation}
T(\mathbf{p};\mathbf{s})=\Tatt(\mathbf{p}; \mathbf{s})+T^{\mathsf{exp}}(\mathbf{p}; \mathbf{s}),
    \label{eq:per_token_latency_general}
\end{equation}
Thus, the per-token inference latency~\eqref{eq:per_token_latency_general} is fully determined by the reduced variable $\mathbf{p}$, which has only $K$ dimensions, one per expert--gateway link.
Since the activation pattern $\mathbf{s}$ is random, we focus on minimizing the average inference latency $\bar T(\mathbf{p})\triangleq\E_{\mathbf{s}}[T(\mathbf{p};\mathbf{s})]$. Under the sufficient condition in~\eqref{eq:near_opt_power_condition}, the reduced problem is given by
\begin{align}
    (\mathcal{P}_1):\quad\min_{\mathbf{p}}\quad & \E_{\mathbf{s}}\!\left[T(\mathbf{p};\mathbf{s})\right] \nonumber\\
    \mathrm{s.t.}\quad & -\eta_k\Ps\le P_k<\Ps,\quad k=1,\ldots,K \nonumber
\end{align}
The constraint follows respectively from the transfer power caps in~\eqref{eq:per_node_tx_caps} and from keeping the expert compute power $P_k^{\mathsf{cmp}}=\Ps-P_k$ positive.

\subsection{Tractable Surrogate of Inference Latency}
\label{sec:surrogate}

After reducing the complexity of the optimization, we now address the second challenge within the intractable problem~$(\mathcal{P}_0)$, resulting in a surrogate function for the average latency.
Firstly, rather than averaging the gateway-side latency over random activations directly, we first average the gateway power and then compute the latency once from this average $\E_{\mathbf{s}}[P_0^{\mathsf{cmp}}]$.
Secondly, since the expert-side latency is determined by the slowest active expert, we approximate it by the average latency over the active experts instead of this maximum.
The resulting surrogate function of the average latency is given in Proposition~\ref{prop:surrogate}.

\begin{proposition}[Surrogate of inference latency]
\label{prop:surrogate}
The average per-token latency $\bar T(\mathbf{p})$ is approximated by the surrogate function, given by
\begin{equation}
    \begin{aligned}
    \Tsur(\mathbf{p})
    &\triangleq \latt\Bigl(\Ps+\textstyle\sum_{k=1}^{K}\bigl[(1-\alpha_k)\eta_k\Ps+\alpha_k\,g_k(P_k)\bigr]\Bigr)^{-\frac{1}{3}}\\
    &\quad+\sum_{k=1}^{K}\frac{\alpha_k}{\kappa}\Bigl[T_k^{\mathsf{cmm}}+\lffn(\Ps-P_k)^{-\frac{1}{3}}\Bigr].
    \end{aligned}
    \label{eq:surrogate_def}
\end{equation}
\end{proposition}

\begin{proof}
    See Appendix~\ref{app:surrogate}.
\end{proof}

The surrogate comprises a gateway-side and an expert-side term, each monotonic in $P_k$ but in opposite directions.
The first term in~\eqref{eq:surrogate_def} represents the gateway attention latency and decreases with $P_k$.
The second term represents the weighted average expert latency and increases with $P_k$ through expert~$k$'s contribution.
The resulting tradeoff is illustrated in Fig.~\ref{fig:gateway_expert_tradeoff} and clarified in the following Remark~\ref{rem:gateway_expert_tradeoff}.

\begin{figure}
    \centering
    \includegraphics[width=0.45\textwidth]{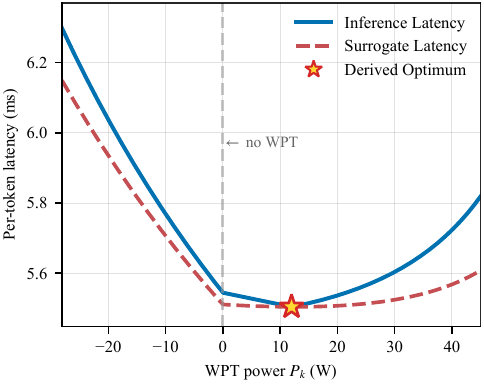}
    \caption{Comparison of the per-token inference latency and its surrogate, illustrating the gateway--expert latency tradeoff.}
    \label{fig:gateway_expert_tradeoff}
\end{figure}

\begin{remark}[Gateway-expert tradeoff]
    \label{rem:gateway_expert_tradeoff}
    The opposing monotonicities of the two terms in~\eqref{eq:surrogate_def} reveal a gateway-expert tradeoff carried by each $P_k$. In particular, raising $P_k$ enlarges the gateway compute power and shortens the gateway-side latency, but it simultaneously shrinks expert~$k$'s compute power and lengthens its expert-side latency.
\end{remark}

Therefore, the optimal power control should balance the tradeoff and can be obtained by solving the following problem.
\begin{align}
(\mathcal{P}_2):\quad\min_{\mathbf{p}}\quad & \Tsur(\mathbf{p}) \nonumber\\
    \mathrm{s.t.}\quad & -\eta_k\Ps\le P_k<\Ps,\quad k=1,\ldots,K, \nonumber\\
    & \Ps+\textstyle\sum_{k=1}^{K}\bigl[(1-\alpha_k)\eta_k\Ps+\alpha_k g_k(P_k)\bigr]>0, \nonumber
\end{align}
where the last constraint keeps the gateway compute power positive.
Since the bounds $-\eta_k\Ps\le P_k<\Ps$ are intervals and the gateway-power constraint keeps a concave function positive (each $g_k$ is concave by~\eqref{eq:bi_ellk}), the feasible set of problem~$(\mathcal{P}_2)$ is convex~\cite{boyd2004convex}.
Meanwhile, the following Lemma~\ref{lem:convexity} establishes the convexity of the objective $\Tsur$, proving that problem $(\mathcal{P}_2)$ is a convex problem.
\begin{lemma}[Convexity of the surrogate]
\label{lem:convexity}
$\Tsur$ is convex on the feasible set of~$(\mathcal{P}_2)$.
\end{lemma}
\begin{proof}
By~\eqref{eq:bi_ellk}, each $g_k$ is concave, so the gateway-side argument $\E_{\mathbf{s}}[P_0^{\mathsf{cmp}}]$ in~\eqref{eq:surrogate_def}, a nonnegative combination of the $g_k$, is concave, and composing it with the convex decreasing map $x\mapsto x^{-1/3}$ yields a convex gateway-side term.
As each expert-side term $\frac{\alpha_k}{\kappa}(\Ps-P_k)^{-1/3}$ is convex on $P_k<\Ps$, $\Tsur$ is a sum of convex terms and hence convex.
\end{proof}

Problem~$(\mathcal{P}_2)$ is thus a tractable convex surrogate of the
original average-latency problem~$(\mathcal{P}_1)$. The next section
derives the optimal solutions to the surrogate problems under
unidirectional and bidirectional WPT. Finally, we bound the
average-latency loss incurred by optimizing the surrogate rather than
the original objective in each WPT regime.

\section{Optimal WPT Power Control}
\label{sec:power_control}

In this section, we solve the surrogate problem~$(\mathcal{P}_2)$ to achieve the near-optimal WPT power control.
We then establish its performance guarantee with respect to the original problem~$(\mathcal{P}_1)$.
Finally, we discuss unidirectional WPT as a practical restriction.

\subsection{Optimal WPT Power Control}

In this subsection, we derive the optimal solution of problem~$(\mathcal{P}_2)$ and analyze its structure and performance.
As established in Lemma~\ref{lem:convexity}, the surrogate objective $\Tsur$ is convex.
At the same time, the feasible set of problem~$(\mathcal{P}_2)$ is also convex.
Hence, problem~$(\mathcal{P}_2)$ is a convex optimization problem.
Its global optimum can therefore be obtained from the KKT conditions, as stated in the following theorem.

\begin{theorem}[Optimal WPT Power Control]
\label{thm:bi_water_filling}
The unique optimal active-state power control for problem~$(\mathcal{P}_2)$
is given by
\begin{equation}
P_k^{\star} =\frac{\Ps}{\nu^{\star}}\Bigl(
  [\nu^{\star}-\chi\eta_k^{-\frac{3}{4}}]^+ -\min\!\left\{
  [\chi\eta_k^{\frac{3}{4}}-\nu^{\star}]^+,
  \eta_k\nu^{\star}
  \right\}\Bigr),
\label{eq:bi_pk_star}
\end{equation}
where $\chi\triangleq\left(\frac{\lffn}{\kappa\latt}\right)^{3/4}$.
The water level $\nu^{\star}>0$ is the unique solution to the following equation
\begin{equation}
\begin{aligned}
&\left[1+\sum_{k=1}^{K}(1-\alpha_k)\eta_k\right]\nu^{\star}
+\sum_{k=1}^{K}\alpha_k\eta_k
 [\nu^{\star}-\chi\eta_k^{-\frac{3}{4}}]^+\\
&\quad-\sum_{k=1}^{K}\frac{\alpha_k}{\eta_k}
 \min\!\left\{
 [\chi\eta_k^{\frac{3}{4}}-\nu^{\star}]^+,
 \eta_k\nu^{\star}
 \right\}=1.
\end{aligned}
\label{eq:bi_water_level_balance}
\end{equation}
\end{theorem}

\begin{proof}
See Appendix~\ref{app:proof_bi_water_filling}.
\end{proof}

Theorem~\ref{thm:bi_water_filling} divides the optimal power control into three modes.
Under the signed-power definition in~\eqref{eq:bi_pk_def}, positive, zero, and negative values of $P_k^{\star}$ correspond to donation, neutrality, and reception, respectively.
To make these three modes explicit, \eqref{eq:bi_pk_star} can be rewritten as
\begin{equation}
P_k^{\star}=
\begin{cases}
\dfrac{\Ps}{\nu^{\star}}
\bigl(\nu^{\star}-\chi\eta_k^{-\frac{3}{4}}\bigr),
& \nu^{\star}>\chi\eta_k^{-\frac{3}{4}},\\[1ex]
0,
& \chi\eta_k^{\frac{3}{4}}\le\nu^{\star}
  \le\chi\eta_k^{-\frac{3}{4}},\\[1ex]
-\dfrac{\Ps}{\nu^{\star}}
\min\!\left\{\chi\eta_k^{\frac{3}{4}}-\nu^{\star},\eta_k\nu^{\star}\right\},
& \nu^{\star}<\chi\eta_k^{\frac{3}{4}},
\end{cases}
\label{eq:bi_threshold_conditions}
\end{equation}
where the $\min\{\cdot\}$ term enforces the receive cap $\eta_k\Ps$ imposed by the gateway transmit-power limit.
The variation of $P_k^{\star}$ with the water level is illustrated in Fig.~\ref{fig:bi_power_transition}.

\begin{figure}[!t]
    \centering
    \includegraphics[width=0.85\columnwidth]{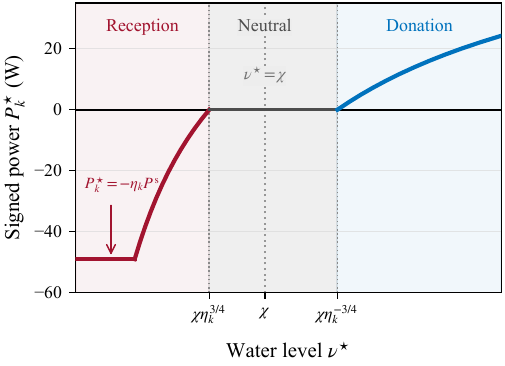}
    \caption{Three transfer modes of the optimal power control.}
    \label{fig:bi_power_transition}
\end{figure}

Figure~\ref{fig:bi_power_transition} shows the mode transition of an expert, driven by the water level $\nu^{\star}$.
As $\nu^{\star}$ increases, the expert moves from reception to neutrality and then to donation.
The reception and donation thresholds lie on opposite sides of $\chi$.
Therefore, the water level cannot activate both transfer directions simultaneously, which is formalized in the following Corollary~\ref{cor:bi_directional_consistency}.

\begin{corollary}[Directionally Consistent Optimal Transfers]
\label{cor:bi_directional_consistency}
All nonzero optimal active-state transfers have the same direction, written as
\[
    P_{k_1}^{\star}P_{k_2}^{\star}\ge0,
    \qquad \forall\, k_1,k_2\in\{1,\ldots,K\}.
\]
\end{corollary}

\begin{figure*}[t]
    \centering
    \includegraphics[width=0.94\textwidth]{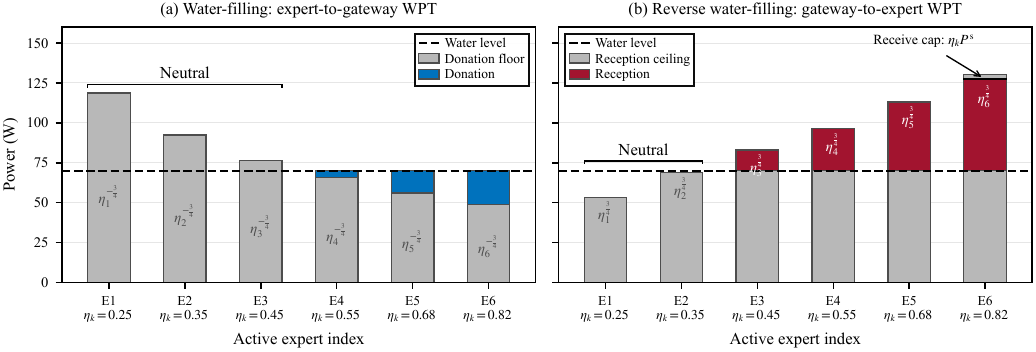}
    \begingroup
    \CLASSOPTIONconferencetrue
    \caption{Two-sided water-filling power control for SpaceWPT.}
    \label{fig:bi_water_filling}
    \label{fig:threshold_partition}
    \label{fig:eta_pattern_effect}
    \endgroup
\end{figure*}

\begin{proof}
By~\eqref{eq:bi_threshold_conditions}, donation by expert~$k_1$ requires
$\nu^{\star}>\chi\eta_{k_1}^{-3/4}\ge\chi$, whereas reception by expert~$k_2$
requires $\nu^{\star}<\chi\eta_{k_2}^{3/4}\le\chi$. Since $\nu^{\star}$ is
shared across all experts, these conditions cannot hold simultaneously. Hence, all
nonzero active-state transfers have the same direction.
\end{proof}

Corollary~\ref{cor:bi_directional_consistency} shows that power donation and reception cannot coexist under the optimal active-state control.
We therefore separate the signed solution in~\eqref{eq:bi_pk_star} into its donation and reception components, written as
\[
\begin{aligned}
(P_k^{\star})^+
&=\frac{\Ps}{\nu^{\star}}\bigl[\nu^{\star}-\chi\eta_k^{-3/4}\bigr]^+,\\
(-P_k^{\star})^+
&=\frac{\Ps}{\nu^{\star}}\min\!\left\{
\bigl[\chi\eta_k^{3/4}-\nu^{\star}\bigr]^+,\eta_k\nu^{\star}\right\}.
\end{aligned}
\]
The two directional components in~\eqref{eq:bi_pk_star} exhibit distinct power-control structures.
For power donation, the optimal power is proportional to the gap between the water level $\nu^{\star}$ and the efficiency-dependent floor $\chi\eta_k^{-3/4}$.
This level-minus-floor relation follows the conventional \emph{water-filling} structure.
For power reception, the optimal power is proportional to the capped gap between the ceiling $\chi\eta_k^{3/4}$ and the water level.
This ceiling-minus-level relation gives rise to a \emph{reverse water-filling} structure.
Together, they form the two-sided water-filling structure illustrated in Fig.~\ref{fig:bi_water_filling} and analyzed in the following Remark~\ref{rem:bi_partition_structure}.

\begin{remark}[Two-Sided Water-Filling Power Control]
\label{rem:bi_partition_structure}
The optimal WPT control extends conventional water-filling to both sides of the water level $\nu^{\star}$.
Expert~$k$ donates power when $\nu^{\star}$ exceeds its floor and receives power when $\nu^{\star}$ falls below its ceiling.
A higher WPT efficiency $\eta_k$ lowers the floor and raises the ceiling, thereby increasing the transferred power between expert $k$ and the gateway.
\end{remark}

\begin{remark}[Efficiency-Dependent Neutral Region]
The reception ceiling and donation floor define a neutral region of width
$\chi(\eta_k^{-3/4}-\eta_k^{3/4})$ for expert~$k$.
Since this region enlarges with $\eta_k$ decreases, experts with lower-efficiency WPT links are more likely to be neutral.
\end{remark}

We now combine the active- and idle-state results to obtain the complete WPT control.
Theorem~\ref{thm:bi_water_filling} gives the optimal signed power for an active expert, while Lemma~\ref{lem:bi_idle_structure} gives the optimal powers for an idle expert.
Using the mapping in~\eqref{eq:bi_active_from_pk}, the optimal WPT powers are given by
\begin{equation}
    \left[P_{k\to0}^{\star}(s_k),P_{0\to k}^{\star}(s_k)\right]
    =
    \begin{cases}
        \left[\Ps,0\right], & s_k=0,\\[0.4ex]
        \bigl[(P_k^{\star})^+,\frac{(-P_k^{\star})^{+}}{\eta_k}\bigr], & s_k=1.
    \end{cases}
    \label{eq:bi_active_star}
\end{equation}

Averaging~\eqref{eq:bi_active_star} over $s_k$ yields the long-term average transfer powers over inference rounds.
The long-term transfer power from expert~$k$ to the gateway is given by
\begin{equation}
    \Pbar_{k\to 0}^{\star}
    \triangleq
    \E_{s_k}\!\left[P_{k\to 0}^{\star}(s_k)\right]
    = (1-\alpha_k)\,\Ps+\alpha_k\,(P_k^{\star})^{+},
    \label{eq:bi_avg_tx_expert_to_center}
\end{equation}
The long-term transfer power from the gateway to expert~$k$ is given by
\begin{equation}
    \Pbar_{0\to k}^{\star}
    \triangleq
    \E_{s_k}\!\left[P_{0\to k}^{\star}(s_k)\right]
    =\alpha_k\,\frac{(-P_k^{\star})^{+}}{\eta_k}.
    \label{eq:bi_avg_tx_center_to_expert}
\end{equation}

\begin{figure*}[!t]
    \centering
    \subfigure[Gateway-side power control]{%
        \includegraphics[width=0.43\textwidth]{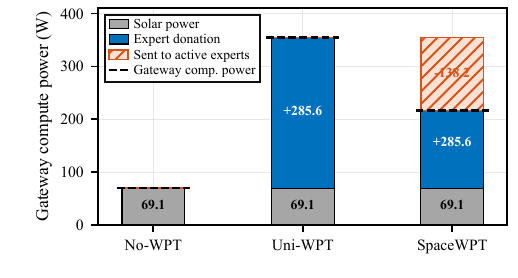}%
        \label{fig:gain_decomposition_power}}
    \hfill
    \subfigure[E2E per-token latency]{%
        \includegraphics[width=0.43\textwidth]{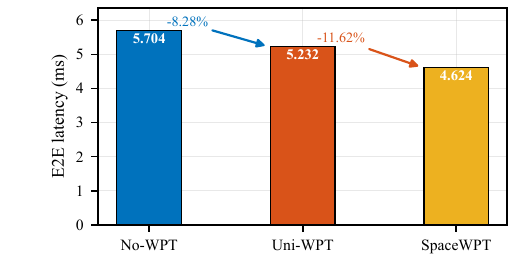}%
        \label{fig:gain_decomposition_latency}}
    \centeredfigurecaption{Power flow and the resulting E2E latency reduction
    for Mixtral.}{fig:gain_decomposition}
\end{figure*}

\subsection{Analysis of the Optimal Solution}

This subsection evaluates the derived WPT control from three perspectives:
solution optimality, complexity, and practical deployment.
We first establish its near-optimality and then quantify its computational advantage over directly solving problem~$(\mathcal{P}_0)$.
Finally, we consider a practical deployment with unidirectional WPT.

We first evaluate $\mathbf{p}^{\star}$ against the optimal solution of the
original problem~$(\mathcal{P}_0)$.
The design principles established in Lemmas~\ref{lem:bi_idle_structure}
and~\ref{lem:no_counter_flow} reduce problem~$(\mathcal{P}_0)$ to
problem~$(\mathcal{P}_1)$ without compromising optimality.
Problem~$(\mathcal{P}_2)$ then replaces the average latency in
problem~$(\mathcal{P}_1)$ with the surrogate $\Tsur$.
The optimality gap between the original one and the transformed one is bounded by the following Corollary~\ref{cor:near_opt}.

\begin{corollary}[Near-optimality bound]
\label{cor:near_opt}
The optimality gap of the solution $\mathbf{p}^{\star}$ satisfies
\begin{equation}
0\le \bar T(\mathbf{p}^{\star})-\min_{\mathbf{p}}\bar T(\mathbf{p})
\le \varepsilon_0,
    \label{eq:near_opt}
\end{equation}
where the constant $\varepsilon_0$ is given in~\eqref{eq:near_opt_constant}
and depends only on the system parameters.
\end{corollary}

\begin{proof}
See Appendix~\ref{app:near_opt}.
\end{proof}

We next compare the computational complexity of the proposed solution
with that of directly solving problem~$(\mathcal{P}_0)$.
\begin{remark}[Computational Complexity]
Directly solving problem~$(\mathcal{P}_0)$ involves $4K$ power variables
and up to $\binom{K}{\kappa}$ activation patterns.
Theorem~\ref{thm:bi_water_filling} reduces the optimization to a
single water-level equation in $\nu^\star$, whose solution determines
all $K$ powers via~\eqref{eq:bi_pk_star}.
This water-level equation can be solved in $\mathcal{O}(K\log K)$ operations.
\end{remark}

Finally, we consider unidirectional WPT as a practical restriction
of the general control.
The gateway executes the attention block for every token, whereas only
$\kappa$ of the $K$ experts are activated.
This workload asymmetry motivates the expert-to-gateway restriction
$P_k\ge0$, which reduces onboard WPT hardware requirements and control
overhead.
The optimal control under this restriction is stated in the
following corollary.

\begin{corollary}[Power Control for Unidirectional-WPT Case]
When $P_k\ge0$, the unique optimal active-state power of expert~$k$ is
\begin{equation}
P_k^{\mathsf{uni}\star}
=
\frac{\Ps}{\nu_{\mathsf{uni}}^\star}
\left[
\nu_{\mathsf{uni}}^\star-\chi\eta_k^{-3/4}
\right]^+,
\label{eq:uni_pk_star}
\end{equation}
where $\nu_{\mathsf{uni}}^\star>0$ is uniquely determined by
\begin{equation}
\begin{aligned}
&\left[1+\sum_{k=1}^{K}(1-\alpha_k)\eta_k\right]
\nu_{\mathsf{uni}}^\star\\
&\quad+\sum_{k=1}^{K}\alpha_k\eta_k
\left[
\nu_{\mathsf{uni}}^\star-\chi\eta_k^{-3/4}
\right]^+
=1.
\end{aligned}
\label{eq:uni_water_level_balance}
\end{equation}
\end{corollary}

The unidirectional control retains the expert-to-gateway water-filling branch
of the general solution, illustrated in Fig.~\ref{fig:bi_water_filling}(a).
We use it as a practical baseline in Section~V to quantify the benefit of
removing the directional restriction.

\section{Experimental Results}
\label{sec:experiments}

\subsection{Experimental Setup}
\label{sec:exp_setup}

\subsubsection{SpaceDC Configuration}

The simulated SpaceDC comprises one gateway satellite and $K$
co-orbiting expert satellites.
Under the HCW relative-motion model in~\eqref{eq:hcw_sol}, expert~$k$
follows a zero-offset closed relative orbit with
$\rho_k=\rho$, $y_k^{\mathsf{off}}=0$, and
$\phi_k=2\pi(k-1)/K$.
The resulting gateway--expert distance varies from $\rho$ to $2\rho$ over
one orbital period.
Unless otherwise specified, we set $\rho=50$\,m and evaluate each scheme at
four uniformly spaced orbital snapshots.

Each gateway--expert pair uses separate FSO and laser-based WPT links for data and power
transfer, respectively.
We set $R(d_k)=100$\,Gbps and
$Q^{\mathsf{dl}}=Q^{\mathsf{ul}}=4096\times16$ bits.
Under the communication model in~\eqref{eq:one_way_tcmm}, this fixed rate
leaves the propagation delay $d_k/c$ as the only distance-dependent
communication term.
The $K$ dedicated WPT links form a star centered at the gateway.
We use a saturated approximation of the end-to-end efficiency
in~\eqref{eq:eta_end_to_end}, given by
$\eta_k=\eta_{\max}\min\{1,(100\,\mathrm{m}/d_k)^2\}$, with
$\eta_{\max}=0.689$ by default.

\subsubsection{GPU Configuration}
The per-satellite solar-power budget is set to $\Ps=69.09$\,W, and every node
is assigned the same GPU profile.
We calibrate the DVFS model in~\eqref{eq:dvfs} using power--frequency
measurements collected on an NVIDIA A100 GPU, with a reference
clock frequency of $990$\,MHz at $69.09$\,W.
The coefficient $\gamma$ incorporates the conversion from GPU clock frequency
to effective computation speed.
This calibration is shared by the gateway and all expert satellites.
Unless otherwise specified, these settings remain fixed, and only the
parameter under study is varied.

\subsubsection{MoE Configuration}

We evaluate Mixtral-8x7B-Instruct and Phi-3.5-MoE-Instruct on the
SpaceDC~\cite{jiang2024mixtral,abdin2024phi3}.
Both models contain $L=32$ MoE layers and adopt native Top-$2$ routing.
Mixtral has $K=8$ experts per layer, whereas Phi has $K=16$.
Both models are evaluated at a $221$-token reference context, obtained by
rounding the mean decoding context of $221.13$ tokens measured over the full
Mixtral--HumanEval workload.

\begin{figure}[!tb]
  \centering
  \makebox[\columnwidth][c]{%
    \subfigure[Mixtral-8x7B]{%
      \includegraphics[width=0.51\columnwidth]{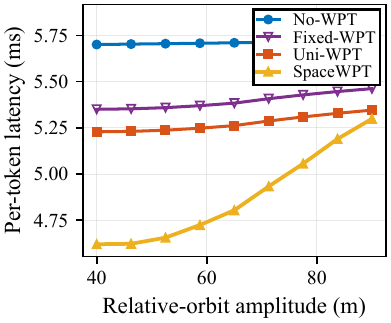}%
      \label{fig:performance_formation_mixtral}}%
    \hspace{0.01\columnwidth}%
    \subfigure[Phi-3.5-MoE]{%
      \includegraphics[width=0.51\columnwidth]{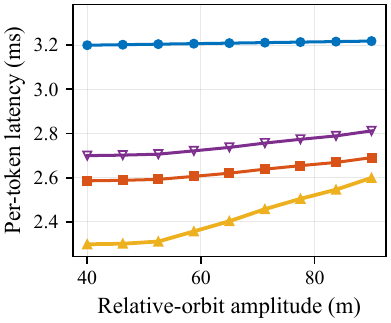}%
      \label{fig:performance_formation_phi}}%
  }
  \caption{E2E latency versus the relative-orbit amplitude of the SpaceDC.}
  \label{fig:performance_formation}
\end{figure}

\begin{figure}[!tb]
  \centering
  \makebox[\columnwidth][c]{%
    \subfigure[Mixtral-8x7B]{%
      \includegraphics[width=0.51\columnwidth]{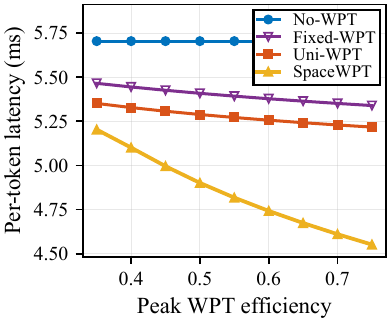}%
      \label{fig:performance_efficiency_mixtral}}%
    \hspace{0.01\columnwidth}%
    \subfigure[Phi-3.5-MoE]{%
      \includegraphics[width=0.51\columnwidth]{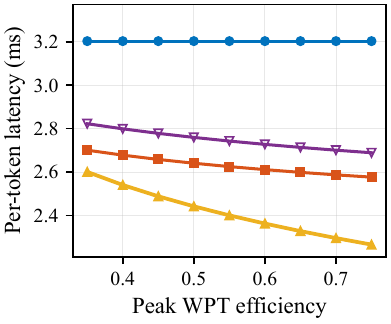}%
      \label{fig:performance_efficiency_phi}}%
  }
  \caption{E2E latency versus peak WPT efficiency.}
  \label{fig:performance_efficiency}
\end{figure}

For each model, we process all $164$ HumanEval prompts using batch size one
and decode $128$ tokens per prompt~\cite{chen2021codex}.
We record the Top-$2$ active expert pair at every MoE layer and decoding step.
Since WPT control changes only the compute power and latency, HumanEval
supplies the activation workload rather than an accuracy metric.

For layer~$\ell$, we write $\alpha_{\ell,k}$ for the empirical marginal
activation probability of expert~$k$, corresponding to $\alpha_k$
in~\eqref{eq:alpha_k_def}.
At every layer and orbital snapshot, we compute the analytical
Uni-WPT and SpaceWPT controls using
$\{\alpha_{\ell,k}\}$ and the snapshot-dependent WPT efficiencies
$\{\eta_k\}$.
We then evaluate all schemes with the recorded joint expert-pair frequencies
using the original average-latency model in~\eqref{eq:avg_latency_physical_def},
which preserves the active-expert maximum.
The reported E2E per-token latency aggregates the modeled latency across layers
and is averaged over expert activations and orbital positions.
Thus, the marginal model is used only to design the controls, whereas the
reported latency retains the observed Top-$2$ activation correlations.

\subsubsection{Benchmarking Schemes}

We compare the proposed analytical controls with a system without WPT and a
fixed-ratio heuristic.
\begin{itemize}
    \item \emph{No-WPT:} Each satellite uses only its local solar-power
    budget for computation.
    \item \emph{Fixed-WPT:} Each idle expert transmits a fixed $50\%$ of
    its solar-power budget to the gateway, whereas active experts perform no
    WPT.
    \item \emph{Uni-WPT:} This unidirectional restriction of SpaceWPT
    allows power to flow only from the experts to the gateway. The active-state powers follow the water-filling control
    in~\eqref{eq:uni_pk_star}, while idle experts follow
    Lemma~\ref{lem:bi_idle_structure}.
    \item \emph{SpaceWPT:} Power may flow in either direction between
    the gateway and each expert, following the two-sided water-filling control
    in Theorem~\ref{thm:bi_water_filling}.
\end{itemize}
All schemes use the same model deployment, routing decisions, orbital
snapshots, and communication and computation settings; they differ only in
the WPT directions and transmit powers.

\subsection{SpaceWPT Performance}
\label{sec:exp_structure}

At the default operating point, Fig.~\ref{fig:gain_decomposition} shows that
the analytical controls reduce the E2E latency of Mixtral in two stages.
No-WPT incurs $5.7041$\,ms/token, whereas Uni-WPT reduces the latency
to $5.2319$\,ms/token, representing an $8.28\%$ improvement.
SpaceWPT further reduces it to $4.6238$\,ms/token, which is $11.62\%$
below Uni-WPT and $18.94\%$ below No-WPT.

The Uni-WPT gain comes entirely from optimal idle-expert donation.
As characterized in Lemma~\ref{lem:bi_idle_structure}, each idle expert
transmits its full solar-power budget, delivering $285.62$\,W on average to
the gateway after WPT loss and raising its compute power from $69.09$ to
$354.71$\,W, as shown in Fig.~\ref{fig:gain_decomposition_power}.
At the default operating point, all active experts remain neutral, with
$P_k^\star=0$, so their computation latency is unchanged from No-WPT.
The resulting $8.28\%$ reduction therefore stems solely from the shorter
gateway attention latency.

SpaceWPT further reduces the latency by reallocating part of the gateway
power to the active experts.
Compared with Uni-WPT, the gateway allocates an average of $138.18$\,W
to active-expert WPT and retains $216.53$\,W for attention computation, as
shown in Fig.~\ref{fig:gain_decomposition_power}.
This reallocation increases the attention-latency component by
$0.1165$\,ms/token but reduces the critical-expert component by
$0.7246$\,ms/token, yielding the net E2E latency reduction of $11.62\%$ in
Fig.~\ref{fig:gain_decomposition_latency}.

\subsection{Effects of System Parameters}
\label{sec:exp_performance}

This subsection compares the four schemes under variations in system parameters.
The evaluation covers WPT-link conditions through the relative-orbit amplitude
and peak WPT efficiency, and onboard computing resources through the
per-satellite solar-power budget and GPU reference frequency.

The gain of SpaceWPT is governed by the per-link WPT
efficiencies, which depend on both the formation size and the peak transfer
efficiency. As shown in Fig.~\ref{fig:performance_formation}, increasing
$\rho$ from $40$ to $90$\,m lengthens the gateway--expert links and reduces
the gain over No-WPT from $18.95\%$ to $7.35\%$ for Mixtral and from
$28.13\%$ to $19.22\%$ for Phi. Conversely,
Fig.~\ref{fig:performance_efficiency} shows that increasing $\eta_{\max}$
from $0.35$ to $0.75$ raises the gain from $8.78\%$ to $20.21\%$ for
Mixtral and from $18.80\%$ to $29.25\%$ for Phi. Both trends follow from
the power delivered for computation: longer links increase transfer loss,
whereas a higher peak efficiency reduces this loss in both WPT directions.
These results indicate that compact formations and high-efficiency laser-based WPT
hardware directly strengthen the latency advantage of SpaceWPT.

\begin{figure}[!tb]
  \centering
  \makebox[\columnwidth][c]{%
    \subfigure[Mixtral-8x7B]{%
      \includegraphics[width=0.51\columnwidth]{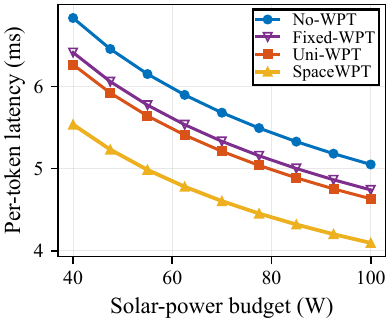}%
      \label{fig:performance_solar_mixtral}}%
    \hspace{0.01\columnwidth}%
    \subfigure[Phi-3.5-MoE]{%
      \includegraphics[width=0.51\columnwidth]{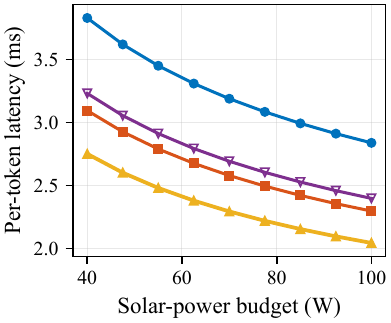}%
      \label{fig:performance_solar_phi}}%
  }
  \caption{E2E latency versus the per-satellite solar-power budget.}
  \label{fig:performance_solar}
\end{figure}

\begin{figure}[!tb]
  \centering
  \makebox[\columnwidth][c]{%
    \subfigure[Mixtral-8x7B]{%
      \includegraphics[width=0.51\columnwidth]{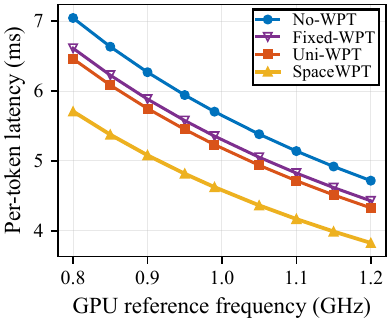}%
      \label{fig:performance_gpu_mixtral}}%
    \hspace{0.01\columnwidth}%
    \subfigure[Phi-3.5-MoE]{%
      \includegraphics[width=0.51\columnwidth]{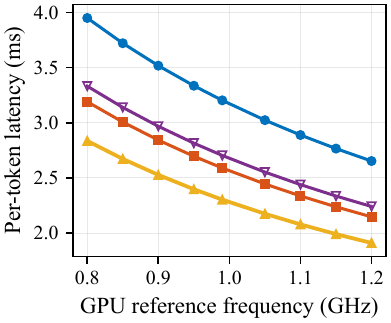}%
      \label{fig:performance_gpu_phi}}%
  }
  \caption{E2E latency versus the GPU reference frequency.}
  \label{fig:performance_gpu}
\end{figure}

Increasing the solar-power budget or GPU frequency reduces computation
latency across the SpaceDC. As shown in Figs.~\ref{fig:performance_solar}
and~\ref{fig:performance_gpu}, this lowers the E2E latency for all four
schemes, while SpaceWPT provides a further reduction throughout both
sweeps---about $19\%$ for Mixtral and $28\%$ for Phi relative to No-WPT. The
two mechanisms are complementary.
Specifically, additional solar power or a higher GPU
frequency accelerates computation at every satellite, whereas WPT
redistributes the available power toward the gateway and active experts
according to instantaneous demand.

\section{Concluding Remarks}
\label{sec:conclusion}

This work advocates inter-satellite WPT to support low-latency distributed inference in SpaceDCs.
The key feature lies in introducing a new degree of freedom for space computing, i.e., WPT-enabled energy migration, to cope with imbalanced computing loads.
Such a power-sharing policy redistributes energy according to demand under spatially uniform sunlight exposure.
The developed principles for migration direction and activation-based WPT control provide theoretical guidance for the development of emerging SpaceDCs in dawn--dusk SSOs.

This work opens a new frontier for integrating WPT into distributed computing in space.
One direction is to extend the WPT design to more general topologies, where direct power exchange between expert satellites allows more flexible power redistribution according to computation demands and link efficiencies.
Another direction is to unify the separate laser-based WPT and communication links into a single optical link that supports both energy delivery and data transmission.
Such integration could reduce onboard hardware requirements while calling for joint designs that balance WPT efficiency and communication performance.

\begingroup
\setlength{\parskip}{3pt plus 1pt minus 1pt}
\setlength{\abovedisplayskip}{10pt plus 2pt minus 2pt}
\setlength{\belowdisplayskip}{10pt plus 2pt minus 2pt}
\setlength{\abovedisplayshortskip}{5pt plus 1pt minus 1pt}
\setlength{\belowdisplayshortskip}{8pt plus 2pt minus 2pt}
\section{Appendix}

\subsection{Proof of \Cref{prop:surrogate}}
\label{app:surrogate}

The surrogate~\eqref{eq:surrogate_def} follows from a certainty-equivalent approximation of the two random components of $\bar T(\mathbf{p})$.

\textit{Gateway-side.} The gateway latency $\latt(P_0^{\mathsf{cmp}})^{-1/3}$ is random through the activation $\mathbf{s}$. Replacing the random power $P_0^{\mathsf{cmp}}$ by its mean inside the nonlinearity yields the certainty-equivalent term $\latt(\E_{\mathbf{s}}[P_0^{\mathsf{cmp}}])^{-1/3}$.
Since expert~$k$ contributes $g_k(P_k)$ to the gateway when active and $\eta_k\Ps$ when idle, with activation probability $\alpha_k$,
\begin{equation*}
\E_{\mathbf{s}}[P_0^{\mathsf{cmp}}]
=\Ps+\textstyle\sum_{k=1}^{K}\bigl[(1-\alpha_k)\eta_k\Ps+\alpha_k g_k(P_k)\bigr],
\end{equation*}
which is the gateway-side term of~\eqref{eq:surrogate_def}.

\textit{Expert-side.} The expert latency is the maximum over the $\kappa$ active experts, $\max_{k\in\mathcal{A}(\mathbf{s})}c_k(P_k)$ with $c_k(P_k)=T_k^{\mathsf{cmm}}+\lffn(\Ps-P_k)^{-1/3}$. Approximating this maximum by the average over the active experts and taking expectation,
\begin{equation*}
\E_{\mathbf{s}}\Bigl[\tfrac{1}{\kappa}\textstyle\sum_{k\in\mathcal{A}(\mathbf{s})}c_k(P_k)\Bigr]
=\frac{1}{\kappa}\sum_{k=1}^{K}\alpha_k\,c_k(P_k),
\end{equation*}
where we used $\Pr\{k\in\mathcal{A}(\mathbf{s})\}=\alpha_k$. This is the expert-side term of~\eqref{eq:surrogate_def}, completing the construction.

\subsection{Proof of \Cref{thm:bi_water_filling}}
\label{app:proof_bi_water_filling}
\label{app:proof_scalar_threshold}

The bidirectional surrogate $\Tsur(\mathbf{p})$ is strictly convex on the
feasible set of~$(\mathcal{P}_2)$, so its KKT point is the unique minimizer.
Define the mean gateway compute power induced by $\mathbf{p}$ as
\begin{equation}
\bar P_0^{\mathsf{cmp}}(\mathbf{p})
\triangleq
\Ps+\sum_{j=1}^{K}
\bigl[(1-\alpha_j)\eta_j\Ps+\alpha_j g_j(P_j)\bigr].
\label{eq:proof_bi_mean_gateway_power}
\end{equation}
At an interior point, stationarity in coordinate~$P_k$ gives
\begin{equation}
\latt\,g_k'(P_k)
\bigl(\bar P_0^{\mathsf{cmp}}(\mathbf{p})\bigr)^{-\frac{4}{3}}
=
\frac{\lffn}{\kappa}(\Ps-P_k)^{-\frac{4}{3}},
\label{eq:proof_bi_stationarity}
\end{equation}
where the common factor $\alpha_k>0$ cancels from both sides.

Let $\bar P_0^{\mathsf{cmp}\star}\triangleq\bar P_0^{\mathsf{cmp}}(\mathbf p^\star)$ and define $\nu^\star\triangleq\Ps/\bar P_0^{\mathsf{cmp}\star}$.
For a donor, $g_k'=\eta_k$, and~\eqref{eq:proof_bi_stationarity} gives
\begin{equation*}
P_k^\star=\frac{\Ps}{\nu^\star}\bigl(\nu^\star-\chi\eta_k^{-3/4}\bigr).
\end{equation*}

For an interior receiver, $g_k'=1/\eta_k$, which instead gives
\begin{equation*}
P_k^\star=-\frac{\Ps}{\nu^\star}\bigl(\chi\eta_k^{3/4}-\nu^\star\bigr).
\end{equation*}

At $P_k^{\star}=0$, the left and right derivatives straddle zero if and
only if
\begin{equation*}
\chi\eta_k^{3/4}\le\nu^{\star}\le\chi\eta_k^{-3/4}.
\end{equation*}

Finally, the lower bound $P_k\ge-\eta_k\Ps$ clips the receiver relation at
$-\eta_k\Ps$.  Combining the donor, neutral, interior-receiver, and
capped-receiver conditions yields~\eqref{eq:bi_pk_star}.

Substituting~\eqref{eq:bi_pk_star} into~\eqref{eq:proof_bi_mean_gateway_power} and dividing by $\bar P_0^{\mathsf{cmp}\star}=\Ps/\nu^{\star}$ gives~\eqref{eq:bi_water_level_balance}. Conversely, any positive solution of this balance equation together with~\eqref{eq:bi_pk_star} satisfies the KKT conditions, including the kink and receive-cap conditions. Existence follows from the optimum of the strictly convex problem, and strict convexity permits only one such solution. Hence both $\nu^{\star}$ and $\mathbf{p}^{\star}$ are unique.\qed

\subsection{Proof of \Cref{cor:near_opt}}
\label{app:near_opt}

For any feasible $\mathbf p$, Jensen's inequality shows that the certainty-equivalent gateway latency is no larger than the corresponding average latency. Moreover, the average latency of the $\kappa$ active experts is no larger than their maximum latency. Therefore,
\begin{equation}
    \Tsur(\mathbf p)\le \bar T(\mathbf p).
    \label{eq:proof_surrogate_lower_bound}
\end{equation}

Let $\eta_{[1]}\ge\cdots\ge\eta_{[K]}$ denote the ordered WPT
efficiencies. The reduced-problem analysis and performance guarantee assume
\begin{equation}
\delta\triangleq1-\kappa+\sum_{k=1}^{K}\eta_k
-\sum_{i=1}^{\kappa}\eta_{[i]}>0
\label{eq:near_opt_power_condition}
\end{equation}
Under the simplified control, this ensures positive gateway power in every state. In fact, $-\Ps\le g_k(P_k)<\eta_k\Ps$, and exactly $\kappa$ experts are active. It follows that
\begin{equation}
P_0^{\mathsf{cmp}}(\mathbf p;\mathbf s)
\ge\Ps\delta>0.
\label{eq:proof_gateway_power_lower}
\end{equation}
Hence, the reduced problems need no additional per-state power constraints.

We next bound the expert compute powers at $\mathbf p^\star$. Let $A\triangleq1+\sum_{j=1}^{K}(1-\alpha_j)\eta_j$.
If $P_k^\star\le0$, then $\Ps-P_k^\star\ge\Ps$. If
$P_k^\star>0$, the directional consistency in
Corollary~\ref{cor:bi_directional_consistency} eliminates the reception
term in~\eqref{eq:bi_water_level_balance}, yielding  $\nu^\star\le A^{-1}.$
Combining this result with the threshold conditions in
\eqref{eq:bi_threshold_conditions} and the receive cap gives
\begin{equation}
\Ps\min\{1,\chi A\eta_k^{-3/4}\}\le\Ps-P_k^\star\le(1+\eta_k)\Ps.
\label{eq:proof_expert_power_bounds}
\end{equation}

Let $\eta_{\max}\triangleq\max_k\eta_k$. The gateway power varies by at
most $\kappa(1+\eta_{\max})\Ps$ over the activation states. Its variance
is thus at most one quarter of the squared range. Applying Taylor's
theorem to $x^{-1/3}$ around the mean gateway power and using
\eqref{eq:proof_gateway_power_lower} bounds the gateway-side
approximation gap.

For the expert-side gap, define
$T_{\max}^{\mathsf{cmm}}\triangleq\max_kT_k^{\mathsf{cmm}}$ and
$T_{\min}^{\mathsf{cmm}}\triangleq\min_kT_k^{\mathsf{cmm}}$.
The difference between the maximum and average latency of any
$\kappa$ active experts is at most $(\kappa-1)/\kappa$ times the global
expert-latency range. Applying~\eqref{eq:proof_expert_power_bounds} to
this range and combining the two approximation gaps gives
\begin{equation}
\bar T(\mathbf p^\star)-\Tsur(\mathbf p^\star)\le\varepsilon_0,
\label{eq:proof_surrogate_gap}
\end{equation}
where
\begin{equation}
\begin{aligned}
\varepsilon_0
&\triangleq
\frac{\latt\kappa^2(1+\eta_{\max})^2}{18\Ps^{1/3}\delta^{7/3}}\\
&\quad+\frac{\kappa-1}{\kappa}
\Bigg[
T_{\max}^{\mathsf{cmm}}-T_{\min}^{\mathsf{cmm}}\\
&\qquad+\lffn\Big(
(P_{\mathrm{e,min}}^{\mathsf{cmp}})^{-1/3}
-[(1+\eta_{\max})\Ps]^{-1/3}
\Big)\Bigg],
\end{aligned}
\label{eq:near_opt_constant}
\end{equation}
where $P_{\mathrm{e,min}}^{\mathsf{cmp}}$ is given by
\begin{equation*}
P_{\mathrm{e,min}}^{\mathsf{cmp}}\triangleq\Ps\min\{1,\chi A\eta_{\max}^{-3/4}\}.
\end{equation*}

Finally, let $\widehat{\mathbf p}$ minimize the original average latency
$\bar T(\mathbf p)$. Since $\mathbf p^\star$ minimizes the surrogate
objective, \eqref{eq:proof_surrogate_lower_bound} and
\eqref{eq:proof_surrogate_gap} give
\begin{equation*}
\begin{aligned}
\bar T(\mathbf p^\star)
&\le \Tsur(\mathbf p^\star)+\varepsilon_0
\le \Tsur(\widehat{\mathbf p})+\varepsilon_0\\
&\le \bar T(\widehat{\mathbf p})+\varepsilon_0.
\end{aligned}
\end{equation*}
Together with
$\bar T(\widehat{\mathbf p})\le\bar T(\mathbf p^\star)$,
this proves~\eqref{eq:near_opt}.
\qed

\endgroup

\bibliographystyle{IEEEtran}
\bibliography{refs}

\end{document}